\documentclass[a4paper,USenglish]{lipics}

\usepackage{microtype}

\usepackage[mathletters]{ucs}
\usepackage[utf8]{inputenc}

\usepackage{amsmath}
\usepackage{bussproofs}
\usepackage{stmaryrd}
\usepackage{comment}
\usepackage{color}
\usepackage{cite}
\usepackage{slashbox}
\usepackage{enumerate}
\usepackage{url}

\usepackage{xspace}
\usepackage{xstring}
\usepackage{ifthen}

\newtheorem{thm}{Theorem}
\newtheorem{exm}{Example}
\newtheorem{dfn}{Definition}
\newtheorem{lem}{Lemma}

\title{Parametricity in an Impredicative Sort}

\author[1]{Chantal Keller}
\author[2]{Marc Lasson}
\affil[1]{INRIA Saclay--Île-de-France at École Polytechnique\\
  \texttt{Chantal.Keller@inria.fr}}
\affil[2]{École Normale Supérieure de Lyon, Université de Lyon, LIP
\footnote{UMR 5668 CNRS ENS Lyon UCBL INRIA}
\\\texttt{marc.lasson@ens-lyon.org}}
\authorrunning{C. Keller and M. Lasson}


\volumeinfo
  {Patrick C\'{e}gielski, Arnaud Durand}
  {2}
  {Computer Science Logic 2012 (CSL'12)}
  {16}
  {1}
  {399}
\EventShortName{CSL'12}
\DOI{10.4230/LIPIcs.CSL.2012.399}

\lstdefinelanguage{Coq}%
  {morekeywords={Variable,Inductive,CoInductive,Fixpoint,CoFixpoint,%
      Definition,Program, Lemma,Theorem,Axiom,Local,Save,Grammar,Syntax,Intro,%
      Trivial,Qed,Intros,Symmetry,Simpl,Rewrite,Apply,Elim,Assumption,%
      Left,Cut,Case,Auto,Unfold,Exact,Right,Hypothesis,Pattern,Destruct,%
      Constructor,Defined,Fix,Record,Proof,Induction,Hints,Exists,let,in,%
      Parameter,Split,Red,Reflexivity,Transitivity,if,then,else,Opaque,%
      Transparent,Inversion,Absurd,Generalize,Mutual,Cases,of,end,Analyze,%
      AutoRewrite,Functional,Scheme,params,Refine,using,Discriminate,Try,%
      Require,Load,Import,Scope,Set,Open,Section,End,match,with,Ltac,fun,forall,exists %
	},%
   sensitive, %
   morecomment=[n]{(*}{*)},%
   morestring=[d]",%
   literate={=>}{{$\Rightarrow$}}1 {>->}{{$\rightarrowtail$}}2{->}{{$\to\,\,\,\,\,$}}1
   {\/\\}{{$\wedge$}}1
   {|-}{{$\vdash$}}1
   {\\\/}{{$\vee$}}1
   {~}{{$\sim$}}1
   {'}{'}1
   {⟦}{{$\llbracket$}}1
   {⟧}{{$\rrbracket$}}1
  }[keywords,comments,strings]%

\DeclareUnicodeCharacter{10214}{\llbracket}
\DeclareUnicodeCharacter{10215}{\rrbracket}

\DeclareMathOperator{\Prop}{\mathtt{Prop}}

\DeclareMathOperator{\Type}{\mathtt{Type}}
\DeclareMathOperator{\Set}{\mathtt{Set}}

\DeclareMathOperator{\Ind}{\mathtt{Ind}}

\DeclareMathOperator{\Boxy}{\mathtt{box}}
\DeclareMathOperator{\NNPP}{\mathtt{Peirce}}
\DeclareMathOperator{\PI}{\mathtt{PI}}

\DeclareMathOperator{\Close}{\mathtt{close}}
\DeclareMathOperator{\Nat}{\mathtt{nat}}
\DeclareMathOperator{\Map}{\mathtt{map}}
\DeclareMathOperator{\Tree}{\mathtt{tree}}
\DeclareMathOperator{\Node}{\mathtt{node}}

\DeclareMathOperator{\Leaf}{\mathtt{leaf}}
\DeclareMathOperator{\Bool}{\mathtt{bool}}
\DeclareMathOperator{\Succ}{\mathtt{S}}
\DeclareMathOperator{\List}{\mathtt{list}}
\DeclareMathOperator{\Church}{\mathtt{church}}
\DeclareMathOperator{\Iter}{\mathtt{iter}}
\DeclareMathOperator{\Vector}{\mathtt{vector}}
\DeclareMathOperator{\nil}{\mathtt{nil}}
\DeclareMathOperator{\cons}{\mathtt{cons}}
\DeclareMathOperator{\Vnil}{\mathtt{Vnil}}
\DeclareMathOperator{\Vcons}{\mathtt{Vcons}}
\DeclareMathOperator{\Land}{\mathtt{and}}
\DeclareMathOperator{\LI}{\mathtt{I}}
\DeclareMathOperator{\False}{\mathtt{False}}
\DeclareMathOperator{\True}{\mathtt{True}}
\DeclareMathOperator{\conj}{\mathtt{conj}}
\DeclareMathOperator{\Or}{\mathtt{or}}
\DeclareMathOperator{\Left}{\mathtt{left}}
\DeclareMathOperator{\Right}{\mathtt{right}}
\DeclareMathOperator{\Even}{\mathtt{even}}
\DeclareMathOperator{\Eq}{\mathtt{eq}}

\DeclareMathOperator{\EqP}{\mathtt{eqP}}
\DeclareMathOperator{\ReflP}{\mathtt{reflP}}
\DeclareMathOperator{\EqT}{\mathtt{eqT}}
\DeclareMathOperator{\ReflT}{\mathtt{reflT}}
\DeclareMathOperator{\Refl}{\mathtt{refl}}

\DeclareMathOperator{\case}{\mathtt{case}}
\DeclareMathOperator{\fix}{\mathtt{fix}}

\DeclareMathOperator{\inv}{\mathtt{inv}}
\DeclareMathOperator{\e}{\mathtt{e}}
\DeclareMathOperator{\fingrp}{\mathtt{fingrp}}
\DeclareMathOperator{\Fingrp}{\mathtt{Fingrp}}
\DeclareMathOperator{\elements}{\mathtt{elements}}

\def\coq{\textsf{Coq}\xspace}
\def\ecc{\textsf{ECC}\xspace}
\def\ml{\textsf{ML}\xspace}
\def\cic{\textsf{CIC}\xspace}
\def\cicr{$\text{\textsf{CIC}}_{\text{\textsf{r}}}$\xspace}
\def\agda{\textsf{Agda}\xspace}
\def\ocaml{\textsf{OCaml}\xspace}
\def\ssreflect{\textsf{Ssreflect}\xspace}
\def\coqparam{\textsf{CoqParam}\xspace}

\newcommand{\ECC}{{\ecc}}
\newcommand{\ML}{{\ml}}
\newcommand{\Fw}{\mathcal{F}_ω}

\newcommand{\rel}[3]{{⟦#1⟧}\,{#2}\,{#3}}

\newcommand{\arrlong}[1]{\overrightarrow{#1}}
\newcommand{\arrvar}[1]{\vec{#1}}

\newcommand{\we}{⊢_{\mathop{\text{SE}}}}

\newcommand\arr[1]{
  \StrLen{$#1$}[\MyStrLen]
  \ifthenelse{\equal{\MyStrLen}{1}}
       {\arrvar{#1}\,\,\!}{\arrlong{#1}}}
\newcommand\arn[2]{
{\arrlong{#1}}^{#2}
}

\begin{document}

\lstset{breaklines=true, xleftmargin=0.3cm, xrightmargin=0.3cm,
  breakatwhitespace=true, mathescape=true, basicstyle=\ttfamily,
  numbers=none, frame=none, language = Coq}

\maketitle

\begin{abstract}
  Reynold's abstraction theorem is now a well-established result for a
  large class of type systems. We propose here a definition of
  relational parametricity and a proof of the abstraction theorem in the
  Calculus of Inductive Constructions (\cic), the underlying formal language of
  \coq, in which parametricity relations' codomain is the impredicative
  sort of propositions. To proceed, we need to refine this calculus by
  splitting the sort hierarchy to separate informative terms from
  non-informative terms. This refinement is very close to \cic, but with
  the property that typing judgments can distinguish informative terms.
  Among many applications, this natural encoding of parametricity inside
  \cic serves both theoretical purposes (proving the independence of
  propositions with respect to the logical system) as well as practical
  aspirations (proving properties of finite algebraic structures). We
  finally discuss how we can simply build, on top of our calculus, a new
  reflexive \coq tactic that constructs proof terms by parametricity.
\end{abstract}

\EnableBpAbbreviations

\section{Introduction}

The \coq system~\cite{Coqdev11} is a proof assistant based on the
Curry-Howard correspondence: propositions are represented as types and
their proofs are their inhabitants. The underlying type system is called
the Calculus of Inductive Constructions (\cic in short). In this type
system, types and their inhabitants are expressions built from the same
grammar and every well-formed expression has a type.

One specificity of \coq among other interactive theorem provers based on Type
Theory is the presence of an impredicative sort to represent the set of
propositions: $\Prop$. Impredicativity means that propositions may be built
by quantification over objects which types inhabit any sort, including the
sort of propositions (for instance the \agda language has a similar type
system except that propositions live in predicative universes
\cite{norell07}). This sort plays a decisive role in the \coq system: in
addition to guaranteeing the compositionality of the propositional world, it
contains the non-computational content, i.e., expressions meant to be erased
by the program extraction process. In particular, it allows the user to add
axioms (like the law of excluded middle, axiom of choice, proof irrelevance,
etc...) without jeopardizing program extraction.

The other sorts are a predicative hierarchy of universes called
$\Type_0, \Type_1, \dots$. Contrary to $\Prop$, it is stratified: one is
not allowed to form a type of a given universe by quantifying over
objects of types of higher universes (stratification has been
introduced in order to overcome Girard's paradox, see
\cite{DBLP:conf/lics/Coquand86} for details). The sort $\Type_0$
(also called $\Set$) contains data-types and basic informative
types. And $\Type_1$ contains types that are quantified over elements of
$\Type_0$, and so on.

One major component of \coq is its extraction
mechanism~\cite{DBLP:conf/types/Letouzey02}, which produces an untyped
term of a \ML-like language from any well typed term of \coq. One
obvious interest is to obtain certified ML code. Roughly speaking, it
proceeds by replacing type annotations and propositional subterms by a
dummy constant. A
difficulty of program extraction is to decide
which terms are informative and which may be erased.
The presence of the sort $\Prop$
only partially solves this problem in \coq since the system has to
distinguish computations over data types from computations over types,
although they all live in $\Type$.

In this paper, we propose a new calculus  which refines the Calculus of
Inductive Constructions, called \cicr. By adding a new predicative hierarchy of sorts
$\Set_0$, $\Set_1$, ..., it confines the types of all informative expressions
and purges the hierarchy $\Type_0$, $\Type_1$, ... of all computational
content. In other words, it guarantees that inhabitants of types in $\Set$
are the only expressions which do not disappear during the extraction
process.

In spite of that, this new calculus may be naturally embedded into
\cic by a very simple forgetful operation. Moreover it remains very
close to \cic and in practice only few terms are not representable
in \cicr. That is why it represents a big step towards an
implementation in the \coq system.

Being able to identify expressions with computational content
-- or in other words \emph{programs} -- was essential to achieve our
initial goal: formalizing the parametricity theory for the Calculus of
Inductive Constructions.

Parametricity is a concept introduced by
Reynolds~\cite{DBLP:conf/ifip/Reynolds83} to study the type abstraction of
System F. It expresses the fact that well-typed programs must behave
uniformly over their arguments with an abstract type: if the type is
abstract, then the functions do not have access to its implementation.
Wadler~\cite{Wadler89} explained how this could be used to
deduce interesting properties shared by all programs of the same type.
Later, Plotkin and Abadi~\cite{DBLP:conf/tlca/PlotkinA93} introduced a logic in which these uniformity
properties can be expressed and proved. This logic may be generalized into
a second-order logic with higher-order individuals
\cite{DBLP:journals/fuin/Takeuti98, DBLP:journals/tcs/Wadler07}.

The main tools of parametricity theory are logical relations defined
inductively over the structure of types together with the so-called
\emph{abstraction theorem}, which builds a proof that any closed program is related to itself for the
relation induced by its type. For instance the relation induced by the type
$∀α,(α → α) → α → α$ of Church numerals is given by the following definition
(represented here in \coq using the standard encoding of relations):
\begin{align*}
 λ(f\,f': ∀α.(α → α) → α → α). ∀ α\,α'\,(R : α → α'→ \Prop)\,(g : α → α) (g' : α' → α').& \\
(∀ x\,x'. R\,x\,x' → R\,(g\,x)\,(g'\,x')) → ∀ z\,z'. R\,z\,z' → R\,(f\,α\,g\,z)\,(f'\,α'\,g'\,z')&
\end{align*}
The abstraction theorem tells that any closed term $F$ of
type $∀α,(α → α) → α → α$ is related to itself according to this relation.

Recently, the work from Bernardy \emph{et
al.}~\cite{DBLP:conf/icfp/BernardyJP10} generalized these constructions up
to a large class of Pure Type Systems and showed that parametricity theory
accommodates well with dependent types. But this cannot be straightforwardly
adapted to \cic, because parametricity relations live in higher universes
instead of using the standard encoding of relations in $\Prop$. Besides,
it is difficult to make parametricity relations live inside $\Prop$ while
conserving abstraction.

But parametricity in a system like \coq would be profitable: it could lead to
more automation, for instance for developing mathematical theories: we
give here an example in finite group theory (Section ~\ref{sec:examples:algebra}). Basing on our refined calculus,
we started an implementation of a \coq tactic that can build closed instances
of the abstraction theorem~\cite{implem12}.

The paper is organized as follows. After explaining in details why we
need to refine the Calculus of Inductive Constructions, we present
\cicr in Sections~\ref{Calculus} and \ref{sec:inductive}.
Section~\ref{sec:parametricity} is devoted to the definition of
relational parametricity, and the proof of the abstraction
theorem, without and with inductive types. In
Section~\ref{sec:examples}, we present different kinds of ``theorems for
free'' that are derived from the general abstraction theorem, like
independence of the law of excluded middle with respect to \cicr or standard
properties of finite groups. We finally
explain the algorithm behind the implementation of the \coq tactic
(Section~\ref{sec:tactic}) before discussing related works and
concluding.

\section{\label{Calculus}\cicr: a refined calculus of constructions with universes}

\subsection{The need for a refinement}

In older versions of \cic, $\Set$ was not a synonym for $\Type_0$ but a
special impredicative sort containing data-types and basic informative
types. However, there is a smaller demand from the users for the
impredicativity of $\Set$ rather than the possibility to add classical
axioms to \cic, and having both may lead to the inconsistency of the
system (the conjunction of excluded middle and description conflicts with
the impredicativity of $\Set$~\cite{Geuvers01}). As a result, nowadays $\Set$
is predicative and behaves in \cic as the first level of the hierarchy of
universes.

In \cicr, we want to reintroduce the sort $\Set$ of
informative types in order to mark the distinction between
expressions with computational content and expressions
which are erased during the extraction process. To stay close to \cic, we
want $\Set$ to be predicative, so we introduce a hierarchy of
sorts $\Set_0$, $\Set_1$, \dots

In the refinement, the \cic hierarchy of sorts $\Type$ is thus divided
into two classes : a hierarchy of sorts $\Set$, whose inhabitants have a
  computational content, and a hierarchy of sorts $\Type$, whose inhabitants are uninformative.
There is a difference of level between inhabitants
of $\Set$ and inhabitants of $\Type$:
the inhabitants of $\Set$ are inhabited only by non-habitable expressions
whereas $\Type$ contains the
signatures of predicates and type constructors which are
themselves, when fully applied, inhabited respectively by proofs and
programs.

In \coq, deciding to which of this two classes an expression of type
$\Type$ belongs is essential in the extraction mechanism.  In the original
two-sorted calculus of constructions (i.e. without universes), the
top-sort contains only arities and therefore the level of terms can be easily
obtained by looking at the type derivation and the extraction procedure is
simple \cite{DBLP:conf/popl/Paulin-Mohring89}. However, in \coq, to extract the
computational content of an inhabitant of sort $\Type$, the extraction
algorithm decides if a type is informative by inspecting the shape of its
normal form \cite{letouzey04, DBLP:conf/cie/Letouzey08}.
Therefore termination of extraction relies on the normalization of
\cic. It makes the correction of the extraction difficult to formally certify.

\subsection{Presentation of the calculus}

The syntax of \cicr is the same as the standard calculus of
constructions except that we extend the set of sorts. Terms are generated by the following grammar:
\newcommand{\pouf}{\hspace{0.8em}}
 $$ A, B \pouf := \pouf  x  \pouf|\pouf  s \pouf|\pouf ∀x:A.B \pouf|\pouf λx:A.B  \pouf|\pouf  (A\,B) $$
where $s$ ranges over the set
 $\left\{\Prop\}∪\{\Set_i, \Type_{i+1} | i ∈ \mathbb{N} \right\}$ of \emph{sorts}
and $x$ ranges over the set of \emph{variables}.
In the remaining of the paper, when no confusion is possible, $\Set$
stands for ``$\Set_i$ for some $i$'', and $\Type$ stands for ``$\Type_i$
for some $i$''.
The notation $i∨j$ represents the maximum of $i$ and $j$.

As usual, we will consider terms up to α-conversion and we denote by
$A[B/x]$ the term built by substituting the term $B$ to each free occurrence of
$x$ in $A$. The $β$-reduction $\rhd$ is defined as in \cic, and we write $A≡B$ to denote the $β$-conversion.

A context $\Gamma$ is a list of couples $x:A$ where $x$ is a variable
and $A$ is term. The empty context is written
$\langle\rangle$. The system has subtyping,
given by the rules in \textsf{\textbf{\small Figure~\ref{fig:subtyping_rules}}}. The typing rules
of \cicr are given in \textsf{\textbf{\small Figure~\ref{fig:refined_calculus}}}.

The word \emph{type} is a synonymous for a term that can be typed by a sort
following those rules. We call \emph{informative types} inhabitants of $\Set$,
\emph{programs} inhabitants of informative types, \emph{propositions} inhabitants
of $\Prop$ and \emph{proofs} inhabitants of propositions.
The sort $\Type_i$ adds a shallow level to the system;
it is populated with two kinds of terms: \emph{arities}, which are terms
 whose head normal forms have the form
  $∀(x₁:A₁)\dots(x_n:A_n).s$ where $s$ is either $\Prop$, $\Set_j$ or
  $\Type_j$ with $j < i$; and
higher-order functions that manipulate arities, and whose types are arities
with $\Type_{i+1}$ as a conclusion. We say that a term has some sort $s$
if $s$ is the type of its type.

\begin{figure*}
\begin{prooftree}
 \AXC{}
 \LeftLabel{$i < j$}
 \RightLabel{\textsc{(Sub$_{1-1}$)}}
 \UIC{$\Set_i <: \Set_j$}

 \AXC{}
 \LeftLabel{$i < j$}
 \RightLabel{\textsc{(Sub$_{1-2}$)}}
 \UIC{$\Type_i <: \Type_j$}

 \noLine
 \BIC{}
\end{prooftree}

\begin{prooftree}
 \AXC{$A <: B$}
 \RightLabel{\textsc{(Sub$_2$)}}
 \UIC{$∀x:C.A <: ∀x:C.B$}
\end{prooftree}
\caption{Subtyping rules}
\label{fig:subtyping_rules}
\end{figure*}

\begin{figure*}
\begin{prooftree}
 \AXC{}
 \RightLabel{\textsc{(Ax$_1$)}}
 \UIC{$⊢ \Prop : \Type_1$}

 \AXC{}
 \RightLabel{\textsc{(Ax$_2$)}}
 \UIC{$⊢ \Set_i : \Type_{i+1}$}

 \AXC{}
 \RightLabel{\textsc{(Ax$_3$)}}
 \UIC{$⊢ \Type_i : \Type_{i+1}$}

 \noLine
 \TIC{}
\end{prooftree}

\begin{prooftree}
 \AXC{$Γ ⊢ A:s$}
 \LeftLabel{$x \not\in Γ, s \in \mathcal{S}$}
 \RightLabel{\textsc{(Var)}}
 \UIC{$Γ,x:A ⊢ x:A$}

 \AXC{$Γ ⊢ B : C$}
 \AXC{$Γ ⊢ A:s$}
 \LeftLabel{$x \not\in Γ, s \in \mathcal{S}$}
 \RightLabel{\textsc{(Weak)}}
 \BIC{$Γ,x:A ⊢ B : C$}

 \noLine
 \BIC{}
\end{prooftree}

\begin{prooftree}
 \AXC{$Γ ⊢  A : C$}
 \AXC{$Γ ⊢  B : s$}
 \LeftLabel{$B ≡ C, s \in \mathcal{S}$}
 \RightLabel{\textsc{(Conv)}}
 \BIC{$Γ ⊢  A : B$}

 \AXC{$Γ ⊢ A : B$}
 \LeftLabel{$B <: C$}
 \RightLabel{\textsc{(Cum)}}
 \UIC{$Γ ⊢ A : C$}

 \noLine
 \BIC{}
\end{prooftree}

\begin{prooftree}
  \AXC{$Γ ⊢  A : r$}
  \AXC{$Γ, x : A ⊢  B : \Set_i$}
  \LeftLabel{$r \in \{\Prop,\Set_i,\Type_i\}$}
  \RightLabel{\textsc{($∀_1$)}}
  \BIC{$Γ ⊢ ∀x:A.B : \Set_i$}
\end{prooftree}

\begin{prooftree}
  \AXC{$Γ ⊢  A : r$}
  \AXC{$Γ, x : A ⊢  B : \Type_i$}
  \LeftLabel{$r \in \{\Prop,\Set_i,\Type_i\}$}
  \RightLabel{\textsc{($∀_2$)}}
  \BIC{$Γ ⊢ ∀x:A.B : \Type_i$}
\end{prooftree}

\begin{prooftree}
  \AXC{$Γ ⊢  A : s$}
  \AXC{$Γ, x : A ⊢  B : \Prop$}
  \LeftLabel{$s \in \mathcal{S}$}
  \RightLabel{\textsc{($∀_3$)}}
  \BIC{$Γ ⊢ ∀x:A.B : \Prop$}
\end{prooftree}

\begin{prooftree}
  \AXC{$Γ ⊢ M : ∀x:A.B$}
  \AXC{$Γ ⊢ N : A$}
  \RightLabel{\textsc{(App)}}
  \BIC{$Γ ⊢ M\,N : B[N/x]$}

  \AXC{$Γ, x : A ⊢  B : C$}
  \RightLabel{\textsc{(Abs)}}
  \UIC{$Γ ⊢ λx :A.B : ∀x:A.C$}
 \noLine
 \BIC{}
\end{prooftree}
\caption{The refined calculus of construction with universes : \cicr}
\label{fig:refined_calculus}
\end{figure*}

\section{Inductive types}\label{sec:inductive}

The calculus is extended with inductive definitions and fixpoints. The presentation is
very similar to Chapter 4.5 of the Reference Manual of
\coq~\cite{Coqdev11}, and one can report to it to have further details.

\subsection{Inductive types and fixpoints}

The grammar of \cicr terms is extended with:
\[
  A, B, P, Q, F  \dots := \cdots \quad  | \quad  I \quad | \quad c \quad
  | \quad \case_I(A,\arr{Q}, P, \arr{F}) \quad | \quad \fix\,(x : A).B
\]

We write $\Ind^p(I:A, c_1:C_1,\dots,c_k:C_k)$ to state that $I$ is a
well-formed inductive definition typed with $p$ parameters,
of arity $A$, with $k$ constructors $c_1,\dots, c_k$ of respective types $C_1,\dots,C_k$.
It requires that:
\begin{enumerate}
\item the names $I$ and $c_j$ are fresh;
\item $A$ is a well-typed arity of conclusion $\Prop$ of $\Set$: it is
  convertible to $∀\arn{(x:P)}{p}\arn{(y:B)}{n}.r$ where $r \in \{\Prop,\Set\}$;
 \item for any $j$, $C_j$ has the form
   $∀\arn{(x:P)}{p}\arn{(z:E_j)}{n_j}.I\,\arn{x}{p}\,\arn{D_j}{n}$ where
   $I$ may appear inside the $E_j$s only as a conclusion. This is called
   the \emph{strict positivity} condition, and is mandatory for the
   system to be coherent~\cite{DBLP:conf/colog/CoquandP88};
 \item for any $j$, $∀\arn{(z:E_j)}{n_j}.I\,\arn{x}{p}\,\arn{D_j}{n}$ is a well-typed expression of sort $r$ under
   the context $(\arn{(x:P)}{p}, I : A)$.
\end{enumerate}
Notice that we do not allow inductive definitions in a nonempty context,
but this is only for a matter of clarity.

Declaring a new inductive definition adds new constants $I$ and $c_j$ to
the system, together with the top left two typing rules presented in
\textsf{\textbf{\small Figure~\ref{fig:inductive_typing_rules}}}.

\begin{figure*}
\begin{prooftree}
  \AXC{}
  \RightLabel{\textsc{(Ind)}}
  \UIC{$⊢ I : A$}

  \AXC{}
  \RightLabel{\textsc{(Constr)}}
  \UIC{$⊢ c_j : C_j$}

  \AXC{$Γ, f : A ⊢ M : A$}
  \LeftLabel{$f$ is guarded}
  \RightLabel{\textsc{(Fix)}}
  \UIC{$Γ ⊢ \fix (f : A).M : A$}

  \noLine
  \TIC{}
\end{prooftree}
\begin{prooftree}
  \AXC{$Γ ⊢ M : I\,\arn{Q}{p}\,\arn{G}{n}$}
  \AXC{$Γ ⊢ T : ∀\arn{y:B[\arn{Q}{p}/\arn{x}{p}]}{n}.I\,\arn{Q}{p}\,\arn{y}{n} → r$}
  \noLine
  \BIC{$\left(Γ ⊢ F_j :
  ∀\arn{(z:E_j[\arn{Q}{p}/\arn{x}{p}])}{n_j}.
                     T\,\arn{D_j[\arn{Q}{p}/\arn{x}{p}]}{n}\,(c_j\,\arn{Q}{p}\arn{z}{n_j})\right)_{j=1\dots k}$}
 \LeftLabel{(under restr.)
 }
  \RightLabel{\textsc{(Case)}}
  \UIC{$Γ ⊢ \case_I(M,\arn{Q}{p},T, \arn{F}{k}) : T\,\arn{G}{n}\,M $}
\end{prooftree}
  \caption{The rules for an inductive type $\Ind^p(I:A, c_1:C_1,\dots,c_k:C_k)$}
  \label{fig:inductive_typing_rules}
\end{figure*}

The bottom rule of \textsf{\textbf{\small Figure~\ref{fig:inductive_typing_rules}}} is the typing
rule for the $\case$ construction which is used to implement elimination
schemes. Two sorts are involved in eliminations: $s$, the sort of the
inductive type we eliminate, which may be $\Prop$ or $\Set$; and $r$,
the type of the type of the term we construct, which may be $\Prop$, $\Set$ or
$\Type$. The four cases of elimination that do not involve $\Type$
are called \emph{small eliminations}.
They are used to build:
\begin{itemize}
\item proofs and programs by inspecting programs;
\item proofs by inspecting proofs;
\item programs by inspecting proofs under restriction~(\ref{eq:restr1})
  (see below).
\end{itemize}
The other two cases are called \emph{large eliminations}. Strong elimination is
mainly used to build propositions by case analysis and to internally
prove the minimality of informative
inductive definitions. For instance, using large elimination over $\Nat$, one may build a
predicate $P$ of type $\Nat → \Prop$ such that $P\,0 ≡ ⊤$ and
$P\,(\Succ\,0) ≡ ⊥$ and thus proves that $0\,≠ (\Succ\,0)$. Similarly,
large elimination may be used to build informative types (for instance, to
build a ``type constructor'' $T_{α,β} :\Nat → \Set$ parametrized by a type
$α$ and an informative type $β$ such that
$T_{α,β}\,n ≡ α → \cdots → α → β$ with $n$ occurrences of $α$)
or to build arities (for instance, to build an ``arity constructor''
$A_α :\Nat → \Type$ parametrized by a type $α$ such that
$A_α\,n ≡α → \cdots → α → \Prop$ with $n$ occurrences of $α$).

Eliminations from $\Prop$ to other sorts are restricted to inductive
definitions that have at most one constructor, and such that all the
arguments (which are not parameters) of this constructor are of
sort $\Prop$:
\begin{equation}\label{eq:restr1}
  k = 0 \text{ or } \left(k = 1 \text{ and } \vdash E : \Prop \text{ for
    any } E ∈ \arn{E_1}{n_1}\right)
\end{equation}
This is essential for coherence~\cite{DBLP:conf/lics/Coquand86}
and has a computational interpretation: it is natural that computing an
informative type should not rely on any proof structure, that would
disappear during program
extraction~\cite{DBLP:conf/types/Letouzey02,letouzey04}.

\begin{exm} Here are a few examples of inductive definitions:
 \begin{itemize}
  \item $\Ind^0 (\Nat : \Set_0, \quad 0 : \Nat, \quad \Succ : \Nat → \Nat)$
  \item $\Ind^1 (\List_i : \Set_i → \Set_i, \quad \nil_i : ∀A:\Set_i.\List_i\,A, \quad \cons_i :∀A:\Set_i. A →  \List_i\,A → \List_i\,A)$
  \item $\Ind^0(\True : \Prop, \quad \LI : \True) $
  \item $\Ind^0(\False : \Prop)$
 \item $\Ind^2(\Eq_i : ∀(A:\Set_i).A → A → \Prop, \quad \Refl_i:  ∀(A:\Set_i)(x:A).\Eq_i\,x\,x)$
 \item $\Ind^2(\EqP : ∀(A:\Prop).A → A → \Prop, \quad \ReflP:  ∀(A:\Prop)(x:A).\EqP \,x\,x)$
 \item $\Ind^2(\EqT_i : ∀(A:\Type_i).A → A → \Prop, \quad \ReflT_i:  ∀(A:\Type_i)(x:A).\EqT_i \,x\,x)$
 \end{itemize}

Note that we have three levels of Leibniz equality: $\Eq_i$ for comparing
programs, $\EqP$ for comparing proofs and $\EqT_i$ for comparing everything
else (we find the same kind of triplication for other standard encodings like
cartesian product, disjoint sum and the existential quantifier).
\end{exm}

The second operator to deal with inductive definitions is fixpoint definition. The typing
rule for the fixpoint is defined on the top right of
\textsf{\textbf{\small Figure~\ref{fig:inductive_typing_rules}}}. It is also restricted to avoid
non-terminating terms, which would lead to absurdity. The restriction is
called the \emph{guard condition}: one argument should have an inductive
type, and must structurally decrease on each recursive call. One may
refer to~\cite{Gimenez95} for further details.

We extend the reduction with the $\iota$-reduction rules:
\begin{equation*}
\begin{array}{rcl}
\case_I(c_j\,\arn{Q}{p}\,\arn{M}{n_j}, \arn{Q}{p}, T, \arn{F}{k})
& \rhd & F_j\,\arn{M}{n_j}\\
(\fix (f:A).M)\,(c_j\,\arn{Q}{p}\,\arn{M}{n_j}) & \rhd & M[\fix (f:A).M/f]\,(c_j\,\arn{Q}{p}\,\arn{M}{n_j})
\end{array}
\end{equation*}
and $\equiv$ denotes the $\beta\iota$-equivalence.

\subsection{Embedding \cicr into \cic and coherence}

This calculus embeds easily into \cic by mapping $\Set_i$ and $\Type_i$
onto the sort $\Type_i$ of \cic:
\begin{lem}
Let $|\bullet|$ be the context-closed function from terms of \cicr to terms
of \cic such that $|\Prop| = \Prop$ and $|\Type_i| = |\Set_i| = \Type_i$, then
we have :
$$Γ ⊢ A : B ⇒  |Γ| ⊢_{\cic} |A| : |B|$$
\end{lem}
Since $|∀X:\Prop.X|=∀X:\Prop.X$, the logical coherence (the existence of an
unprovable proposition) of \cic ensures the coherence of \cicr.

Conversely, some terms of \cic do not have a counterpart in the
refinement: we cannot mix informative and uninformative types. An
example is the following \coq definition:
\begin{lstlisting}
fun (b:bool) => if b then nat else Set
\end{lstlisting}

\section{Relational parametricity}\label{sec:parametricity}

In this setting, we have a natural notion of parametricity: we can
define a translation that maps types to relations and other terms to
proofs that they belong to those relations. What is new is that
relations over objects of type $\Prop$ or $\Set$ have $\Prop$ as a
codomain, which is more natural in a calculus with an impredicative sort
for propositions.

We go step by step. First, we define parametricity for the calculus
without inductive types, and show the abstraction theorem for this
restriction. Subsequently, we add inductive types with large
eliminations forbidden, and finally see how large eliminations behave
with parametricity.

\subsection{Parametricity for the calculus without inductive types}
\setcounter{equation}{0} 

\begin{dfn}[\label{Parametricity}Parametricity relation]
The parametricity translation $⟦\bullet⟧$ is defined by induction on the
structure of terms:
\begin{align}
  ⟦\langle\rangle⟧ = &\,\langle \rangle \\
  ⟦Γ, x:A⟧ = &\,⟦Γ⟧, x:A, x':A',x_R :⟦A⟧\,x\,x' \\
  \nonumber\\
  ⟦s⟧ = &\,λ(x:s)(x':s).x → x' → \hat{s} \\
  ⟦x⟧ = &\,x_R \\
  \nonumber
  ⟦∀x\!:\! A.B⟧ = &\,λ(f:∀x:A.B)(f':∀x':A'.B').
  \quad ∀(x:A)(x':A')(x_R:\rel{A}{x}{x'}).\\
  & \qquad \rel{B}{(f\,x)}{(f'\,x')} \label{DefProd} \\
  ⟦λx:A.B⟧ = &\,λ(x:A)(x':A')(x_R:\rel{A}{x}{x'}).⟦B⟧ \\
  ⟦(A\,B)⟧ = &\,(⟦A⟧\,B\,B'\,⟦B⟧)
\end{align}
with $\hat{\Prop} = \hat{\Set_i} = \Prop$ and $\hat{\Type_i} = \Type_i$
and where $A'$ denotes the term $A$ in which we have replaced each variable
$x$ by a fresh variable $x'$.
\end{dfn}

It is easy to prove by induction that the previous definition is
well-behaved with respect to substitution and conversion:

\begin{lem}[Substitution lemmas]\label{substLemma}
\begin{enumerate}
\item $\left(A[B/x]\right)' = \,A'[B'/x']$
\item $⟦A[B/x]⟧ = \,⟦A⟧[B/x][B'/x']\left[⟦B⟧/x_R\right]$
\item $A₁ ≡_{\beta} A₂ ⇒ \,⟦A₁⟧ \equiv_{\beta} ⟦A₂⟧$
\end{enumerate}
\end{lem}

The abstraction theorem states that the parametricity transformation
preserves typing.

\begin{thm}[\label{Abstraction}Abstraction without inductive definitions]
If $Γ ⊢ A : B$, then $⟦Γ⟧ ⊢ A : B$, $⟦Γ⟧ ⊢ A' : B'$, and
    $⟦Γ⟧ ⊢ ⟦A⟧ : \rel{B}{A}{A'}$.
\end{thm}

\begin{proof}
  The proof is a straightforward
  induction
  on the derivation of $\Gamma \vdash A : B$. The first item is essentially
  proved by invoking structural rules and by propagating induction hypothesis.
  The key steps of the second items are the rule (\textsc{Ax$_2$}), which requires cumulativity, and
  the rules (\textsc{$\forall_1$}-\textsc{$\forall_3$}), which involve
  many abstraction and product rules.
\end{proof}

\subsection{\label{WhyDoesNotWork}Why does not it work directly in \cic ?}

In the syntactic theory of parametricity for dependent types presented in
\cite{DBLP:conf/icfp/BernardyJP10}, relations over a type of some universe
are implemented as predicates ranging in the same universe. This can be read
in the following piece of definition : $⟦\Type_i⟧ = λ(x\,x' : \Type_i).x → x' → \Type_i$. We cannot simply replace the conclusion with $\Prop$, because in
$\cic$ one has $⊢ \Type_i : \Type_{i+1}$, and the abstraction theorem
would require that $⊢ ⟦\Type_i⟧ : ⟦\Type_{i+1}⟧\,\Type_i\,\Type_{i}$
which is equivalent to
 $⊢  λ(x\,x' : \Type_i).x → x' → \Prop : \Type_i\, → \Type_{i} → \Prop$
but this last sequent is not derivable. In our refinement, $⟦\Prop⟧$ and $⟦\Set⟧$ have $\Prop$ as a conclusion,
but this is not a problem since we do not have $⊢ \Set_i : \Set_{i+1}$.

This refined calculus is very convenient to set the basis for
parametricity. As we argued, it has also nice properties regarding
realizability and extraction: as an example, the correctness of
extraction in this calculus would not rely on the termination of the $\beta$-reduction.
Even if possible, obtaining the same result directly in \cic would have
required a complete reworking of parametricity relations.

The calculus is very close to \cic, though. In Section~\ref{sec:tactic},
we discuss if it is possible to write a tactic in \coq that would
exploit this work, without changing \coq's calculus.

\subsection{Adding inductive types}

As a first step, we restrict ourselves to small eliminations: we do not
allow large eliminations. We will see in Subsection~\ref{OvercomeWE}
that we are actually able to handle large eliminations over a big
class of inductive definitions.

We write $Γ \we A : B$ to denote sequents typable in \cicr where large
eliminations are forbidden. Let us suppose that $\Ind^p(I:A, \arn{c :
  C}{k})$, we will define a fresh inductive symbol $⟦I⟧$ and a
family $(⟦c_i⟧)_{i=1...k}$ of fresh constructor names. Then we extend
\textsf{Definition \ref{Parametricity}} with
 \begin{align*}
   ⟦\fix(x:A).B⟧ = & \left(\fix(x_R:⟦A⟧\,x\,x').⟦B⟧\right) [\fix(x:A).B/x][\fix(x':A').B'/x']\\
   ⟦\case_I(M,\arn{Q}{p}, T,\arn{F}{n})⟧ = & \case_{⟦I⟧}(⟦M⟧,\arn{Q, Q',
     ⟦Q⟧}{p}, Θ_I(\arn{Q}{p},T,\arn{F}{n}),\arn{⟦F⟧}{n})
\end{align*}
where $Θ_I$ is defined in \textsf{\textbf{\small Figure~\ref{defTheta}}}.

\begin{figure*}
\begin{align*}
  Θ_I(\arn{Q}{p},T,\arn{F}{n}) = &\,λ\arn{(x:A)(x':A')(x_R:⟦A⟧\,x\,x')}{n}\,(a : I\,\arn{Q}{p}\,\arn{x}{n})
                 (a': I\,\arn{Q'}{p}\,\arn{x'}{n})\\
& \quad (a_R : ⟦I⟧\,\arn{Q\,Q'\,⟦Q⟧}{p}\,\arn{x\,x'\,x_R}{n} a\,a').\\
      & ⟦T⟧\,\arn{x\,x'\,x_R}{n}\,a\,a'\,a_R\,
                           \,(\case_I\,(a, \arn{Q}{p}, T, \arn{F}{n}))
                           \,(\case_{I}\,(a', \arn{Q'}{p}, T', \arn{F'}{n}))
\end{align*}
\caption{\label{defTheta} The definition of $Θ_I$}
\end{figure*}

We want to extend \textsf{Theorem \ref{Abstraction}}
with inductive definitions. We prove the following theorem:

\begin{thm}[\label{AbstractionInductive}Abstraction with inductive definitions]
 \begin{enumerate}
  \item  If $\Ind^p(I:A, \arn{c : C}{k})$ is a valid inductive definition
       then so is $\Ind^{3p}(⟦I⟧ : ⟦A⟧\,I\,I, \arn{⟦c⟧ : ⟦C⟧\,c\,c'}{k})$.
  \item If $Γ \we A : B$ then $⟦Γ⟧ \we A : B$, $⟦Γ⟧ \we A' : B'$, and
      $⟦Γ⟧ \we ⟦A⟧ : \rel{B}{A}{A'}$.
 \end{enumerate}
\end{thm}

\begin{proof}
  The first item requires
  to check the constraints to build inductive types: the typing and the
  strict positivity. As for \textsf{Theorem \ref{Abstraction}}, the second item
  is proved by induction on the structure of the proof of $\Gamma \vdash
  A : B$. One needs to check that the guard condition is preserved in
  the (\textsc{Fix}) rule and that the (\textsc{Case}) rule is
  well-formed. The key idea here is that the translation of terms
  containing only small eliminations also contains only small
  eliminations.
\end{proof}

\subsection{\label{OvercomeWE}Overcoming the restriction over large
  elimination}

Suppose we now authorize the whole large elimination (with
restriction~(\ref{eq:restr1})). The definition generated by the following
inductive definition $\Ind^0(\Boxy_i : \Set_{i+1}, \Close_i : \Set_i → \Boxy_i)$ is
\begin{align*}
\Ind^0\big(&⟦\Boxy_i⟧ : \Boxy_i → \Boxy_i → \Prop, \\
       & ⟦\Close_i⟧ : ∀(A\,A':\Set_i).(A → A' → \Prop) →
 ⟦\Boxy_i⟧\,(\Close_i\,A)\,(\Close_i\,A')\big)
\end{align*}
If we want to prove parametricity for the (\textsc{Case}) rule when we
build a $\Type$, one should provide an inhabitant of:
$∀(A\,A':\Set_i).⟦\Boxy_i⟧\,(\Close_i\,A)\,(\Close_i\,A') → (A → A' → \Prop)$.
But since $⟦\Boxy_i⟧\,(\Close_i\,A)\,(\Close_i\,A')$ has type $\Prop$ and $A →
A' → \Prop$ has type $\Type$, we cannot build the expected relation by
deconstructing a proof of $⟦\Boxy_i⟧\,(\Close_i\,A)\,(\Close_i\,A')$: this is
forbidden by restriction~(\ref{eq:restr1}).

However, let us consider the following example:
$$\Ind^0 (I : \Set, \mathrm{N} : \Nat \to I, \mathrm{B} : \Bool \to I)$$
Let say we need to translate the following large elimination (for the
sake of readability, we present it with the \coq syntax):
\begin{lstlisting}
Definition f (x : I) := match x with
  | N n => vector n
  | B b => nat
end.
\end{lstlisting}
We can swap the destruction of $x_R$ for two nested destructions of
$x$ and $x'$ which produces $k^2$ branches (where $k$ in the number of constructors). But only $k$ of them are actually possible
(we use here the \texttt{Program} keyword in order to let the system infer dependent type annotations for each \texttt{match}):
\begin{lstlisting}
Program Definition f_R (x x' : I) (x_R : ⟦I⟧ x x') :=
  match x with
  | N n => match x' with
    | N n' => let n_R := inv n n' x_R in ⟦vector n⟧
    | B b' => absurd (vector n ->nat ->Prop) (abs$_{12}$ n b' x_R)
    end
  | B b => match x' with
    | N n' => absurd (nat ->vector n' ->Prop) (abs$_{21}$ b n' x_R)
    | B b => ⟦nat⟧
    end
  end.
\end{lstlisting}
where the following terms are all implemented with an authorized large
elimination:
\begin{eqnarray*}
  \texttt{inv} & :& ∀ (n\,n' : \Nat). ⟦I⟧\,(\mathrm{N}\,n)\,(\mathrm{N}\,n') → ⟦\Nat⟧\,n\,n'\\
  \texttt{abs}_{12} & :&  ∀ (n : \Nat) (b' : \Bool). ⟦I⟧\,(\mathrm{N}\,n)\,(\mathrm{B}\, b') → \False \\
  \texttt{abs}_{21} & :&  ∀ (b : \Bool) (n' : \Nat).  ⟦I⟧\,(\mathrm{B}\,b)\,(\mathrm{N}\, n') → \False \\
  \texttt{absurd} & :& ∀ (α : \Type).\False → α
\end{eqnarray*}
We notice that this example runs smoothly because all the arguments of
all the constructors have type $\Prop$ or $\Set$, which avoids the
pitfall of the $\Boxy$ example.

That is why we propose to restrict large elimination from $\Set$ to
$\Type$ to the class of small inductive definitions (this class was
introduced by Paulin in \cite{springerlink:10.1007/BFb0037116} to
restrict the large elimination in vanilla \coq where the sort $\Set$ of
informative types is impredicative):
\begin{dfn}[Small inductive definitions]
 We say that $\Ind^p(I:A, \arn{c : C}{k})$ is a \emph{small
 inductive definition} if all the arguments of each constructor
 are of sort $\Prop$ or $\Set_m$ for some $m$. More formally,
 if for all $1 \leq i \leq k$,
$⊢ c_i : ∀\arn{(x:P)}{p}\arn{(y:B)}{n_i}.I\,\arn{x}{p}\,\arn{D_j}{n}$
 then $\arn{x:P}{p},\arn{y:B}{j-1} ⊢ B_j : r$ with $r = \Prop$ or
 $r = \Set_m$ for some $m$.
\end{dfn}

With this restriction, the abstraction theorem holds in presence of
large elimination:
\begin{thm}
  $\textsf{Theorem~\ref{AbstractionInductive}}$ holds
  when $\we$ stands for derivability where large elimination
  is authorized over small inductive definitions and forbidden
  otherwise.
\end{thm}

\section{Examples of ``free theorems''}\label{sec:examples}

In this section we give a few examples of consequences of the abstraction
theorem. Most examples that can be found in the literature (see for
instance \cite{Wadler89, DBLP:conf/icfp/BernardyJP10})
may be easily implemented in our framework.
To improve readability, we use "$=_α$" and "$∃x:α$" to denote respectively
standard inductive encodings of the Leibniz equality and existential quantifier.

\subsection{The type of Church numerals}\label{sec:examples:church}

Let $\Church_i$ be $∀α:\Set_i, (α → α) → α → α$, the type
of Church numerals. Let $\Iter_i$ be the following expression
\begin{align*}
 \fix\,\Iter_i : \Nat → \Church_i.
& λ (n : \Nat)(α : \Set_i)(f : α → α)(z: α). \\
 & \case(n,λk:\Nat.α,z,λp:\Nat.f\,(\Iter_i\,p\,α\,f\,z))
\end{align*}
which is the primitive recursive operator which composes a function
$n$ times with itself.

The relation $⟦\Church_i⟧ : \Church_i → \Church_i → \Prop$ is the
relation unfolded in the introduction. One can prove easily the following
property on any $f : \Church_i$:
\[
  ⟦\Church_i⟧\,f\,f →  ∃ n : \Nat.
  ∀ (α : \Set_i) (g : α → α) (z:α). \Iter_i\,n\,α\,g\,z =_α f\,α\,g\,z
\]
which states that, if $f$ is in relation with itself by $⟦\Church_i⟧$,
then there exists an integer $n$ such that $f$ is
extensionally equal to $\Iter_i\,n$. Now suppose we have a
closed term $F$ such that $⊢ F : \Church_i$. By the abstraction theorem
we obtain a proof $⟦F⟧$ that $⟦\Church_i⟧\,F\,F$ and therefore that $F$
is extensionally equal to $\Iter_i\,n$ for some $n$.

\subsection{The tree monad}\label{sec:examples:tree}

Binary trees carrying information of type $α$ on their leaves may be
implemented by the following inductive definition :
$$\Ind^1(\Tree_i : \Set_i → \Set_{i+1},
       \Leaf_i : ∀α:\Set_i. α → T\,α,
       \Node_i : ∀α:\Set_i. T\,α → T\,α → T\,α)$$
and it is possible to represent in \cic the function $\Map_i$ of type
$∀(α\,β:\Set_i). (α → β) → \Tree_i\,α → \Tree_i\,β$ which maps a
function to all the leaves of a tree.

The generated relation $⟦\Tree_i⟧$
tells that two trees are related if they have the same shape and
elements at the same position in each tree are related. It is then not
difficult to prove for any function $f : α → α'$ that $⟦\Tree_i⟧\,α\,α'\,R_f$
is a relation representing the graph of the $\Map$ function where
$R_f$ is $λ(x:α)(x':α').f x =_{α'} x'$ and represents the graph of $f$.

We can also define in the system the multiplication of the monad by
programming an expression $\mu_i$ of type $∀α.\Tree_i\,(\Tree_i\,α) →  \Tree_i α$ with the following computational behavior:
 $$ \mu_i\,α\,(\Leaf_i\,α\,x) ≡  \,x \hspace{2em}\text{ and }\hspace{2em}
  \mu_i\,α\,(\Node_i\,α\,x\,y) ≡  \,\Node_i\,α\,(\mu_i\,α\,x)\,(\mu_i\,α\,y)$$
As $\mu_i$ is closed, an application of the abstraction theorem which
instantiates the relation to the graph of $f$ proves the naturality of
$\mu_i$.

\subsection{Parametricity and algebra}\label{sec:examples:algebra}

Obtaining ``free theorems'' by parametricity can be extended to data
types with structure. In this section, we take the example of finite
groups, which is directly related to the \ssreflect
library~\cite{DBLP:conf/tphol/GonthierMRTT07} developed in \coq; but our
reasoning applies to a large variety of algebraic structures.

In Chapter 3.4 of his PhD. thesis~\cite{Garillot11}, François Garillot
observed that algebraic developments require lots of proofs by
isomorphism, which often look similar. Intuitively, a polymorphic function
operating on groups can only compose elements using the laws given by
the group's structure, and thus cannot create new elements.

More formally, we take an arbitrary group $\mathcal{H}$ defined by a
carrier $\alpha : \Set_0$, a unit element $\e : \alpha$, a composition
law $\cdot : \alpha → \alpha → \alpha$, an inverse function $\inv :
\alpha → \alpha$, and the standard axioms stating that $\cdot$ is
associative, $\e$ is neutral on the left and composing with the inverse
on the left produces the unit. On top of this, we define the type of all
the finite subgroups of $\mathcal{H}$ with the following one-constructor
inductive definition:
\begin{equation*} \begin{array}{l}
\Ind^0 \Big(\fingrp : \Set_0, \Fingrp : ∀ \elements: \List \alpha. \\
\qquad \qquad \e ∈ \elements →\\
\qquad \qquad (∀ x\,y. x ∈ \elements → y ∈ \elements → x \cdot y ∈ \elements) → \\
\qquad \qquad (∀ x. x ∈ \elements → \inv x ∈ \elements) → \fingrp\Big)
\end{array}
\end{equation*}
where $∈ : \alpha → \List \alpha → \Prop$ is the standard inductive
predicate stating if an element appears in a list.

Suppose we have a closed term $Z : \fingrp → \fingrp$ (examples of such
terms abound: eg. the center, the normalizer, the derived
subgroup\dots). The abstraction
theorem states that for any $R : \alpha → \alpha → \Prop$ compatible
with the laws of $\mathcal{H}$ and for any $G\,G' : \fingrp$, $⟦\fingrp⟧_R\,G
\,G' → ⟦\fingrp⟧_R\,(Z\,G)\,(Z\,G')$ where $⟦\fingrp⟧_R$ is the relation
on subgroups induced by $R$. Given this, we can prove the following
properties:
\begin{itemize}
\item for any $G$, $Z\,G \subset G$ (if we take $R : x\,y \mapsto x \in G$);
\item for any $G$, for any $\phi$ a morphism of $\mathcal{H}$,
  $\phi(Z\,G) = Z\,\phi(G)$ (if we take $R : x\,y \mapsto y = \phi(x)$).
  It entails that $Z\,G$ is a \emph{characteristic subgroup} of
  $\mathcal{H}$.
\end{itemize}
To prove this, we use the axiom of proof irrelevance (that can be safely
added to the system as we will show in the next subsection). The proof
is straightforward by unfolding the definitions. A complete \coq script
can be found online~\cite{implem12}.

\subsection{Classical axioms}\label{sec:examples:axioms}

One interesting feature of \coq is the ability to add axioms in the
system. However when the parametricity transformation $⟦\cdot⟧$ will
encounter the axiom, it will ask for a proof that it is related to
itself. Let consider an axiom $P$ such that $⊢ P : s$ where $s$ is $\Prop$ or $\Set$.
Here three situations are possible:
\begin{itemize}
  \item Either $P$ is what we call \emph{provably parametric}: the user
   can provide a proof of $∀ h : P. ⟦P⟧\,h\,h$ and this proof may be used
   by the abstraction theorem to prove parametricity for terms involving
   the axiom.
  \item Or $P$ is \emph{provably not parametric}: there exists a proof
   that $∀ (h\,h' : P). ¬(⟦P⟧\,h\,h')$. It means that the axiom would break the
   parametricity of the system: there is no way to invoke the abstraction
   theorem on a term which uses that axiom.
  \item Or it is neither provably parametric nor provably not parametric
   or the user does not know. In this case, the parametricity of the
   axiom may be added as a new axiom at the user's risk.
\end{itemize}
Note that if $¬P$ is provable then $P$ is both provably parametric and
provably not parametric and by the abstraction theorem, if $P$ is provable
then it is of course provably parametric. It is also easy to deduce from the
abstraction theorem that if $P → Q$ is provable then
 $P$ provably parametric implies $Q$ provably parametric,
    and $Q$ provably not parametric implies $P$ provably not parametric.
Hence these notions do not depend on the formulation of your axioms.

\subsubsection{Proof irrelevance}\label{sec:examples:axioms:pi}

The axiom of \emph{proof irrelevance}
$\PI =  ∀ (X:\Prop)(p\,q : X), p =_X q$
states that there is at most one proof
of any proposition. It is provably parametric since
\begin{align*}
&⟦\PI⟧\,h\,h' =  ∀ (X\,X':\Prop)\,(X_R:X→ X'→ \Prop) \\
&\qquad  (p:X)(p':X')(p_R : X_R\,p\,p')
 (q:X)(q':X')(q_R : X_R\,q\,q'). ⟦\EqP⟧\,X\,X'\,p\,p'\,p_R\,q\,q'\,q_R
\end{align*}
may be proved (with $\PI$) equivalent to
\begin{align*}
& ∀ (X\,X':\Prop)\,(X_R:X→ X'→ \Prop)
  (p:X)(p':X')(p_R : X_R\,p\,p'). \\
& \qquad \quad ⟦\EqP⟧\,X\,X'\,p\,p'\,p_R\,p\,p'\,p_R
\end{align*}
which is directly provable by $⟦\ReflP⟧$.
Therefore $\PI$ may be safely added to the system.

\subsubsection{Independence of the law of excluded middle}\label{sec:examples:axioms:independence}

From a user perspective provably not parametric axioms are bad news, but
it provides meta-theoreticians a very simple way to prove independence results.
Indeed, if a formula is provably not parametric then the abstraction theorem
tells you this formula is not provable without large elimination over
not small inductive definitions.

\begin{lem}
  If $P$ is provably not parametric, there is no closed term
  $A$ of type $P$ (in the restriction of large elimination to small
  inductive definitions).
\end{lem}

For instance, Peirce's law  $ \NNPP =  ∀ (X\,Y:\Prop). ((X → Y ) → X)  → X$
(which is known to be equivalent to the
 excluded middle) is provably not parametric.

\section{Towards a \coq implementation}\label{sec:tactic}

This paper sets the theoretical foundation for an implementation of a
reflexive \coq tactic generating the consequences of parametricity for
definitions in the Calculus of Constructions. Two approaches are
possible:
\begin{itemize}
\item modify \coq's calculus to implement \cicr. The implementation of
  the translation becomes straightforward;
\item do not modify \coq's calculus, but let the translation distinguish
  informative terms.
\end{itemize}
The first approach would require to transform \coq radically. We followed
the second approach, and started the implementation of a prototype for
\coq commands and tactics for parametricity, called \coqparam
\cite{implem12}.

In a system like \coq, \emph{reflection} establishes a correspondence
between:
\begin{itemize}
\item a subset of the \coq terms: this is called the \emph{shallow
    embedding};
\item a \coq inductive data type representing these terms: this is
  called the \emph{deep embedding};
\item the \ocaml internal representation of those terms.
\end{itemize}
The deep embedding and the \ocaml representation give access to the
structure of the terms (whereas the shallow embedding does not), which
is very useful to build properties and proofs by computing over this
structure. This process, called \emph{computational reflection}, is a
well-known way to design powerful automatic tactics in
\coq~\cite{DBLP:conf/tphol/GregoireM05,GregoireTW06,DBLP:conf/cpp/ArmandFGKTW11}.

Parametricity comes well within the spirit of computational reflection:
the abstraction theorem is a way to build proofs of terms by inspecting
their structures. Our tactic is based on this remark: given a well-typed
closed term $\vdash A : B$, it builds the well-typed proof $\vdash
\llbracket A \rrbracket : \llbracket B \rrbracket\,A\,A$, going from the
shallow embedding to the \ocaml internal representation (this step is
called \emph{reification}), and the other way round. The difficulty is
to decide, during reification, whether objects of type $\Type$ in \coq
should have type $\Set$ or $\Type$ in \cicr. The tactic does not handle
this yet (as well as full inductive types).

Notice that, with this method, we do not have to generally prove the
abstraction theorem in \coq: \coq's type checker will prove it on each
instance. One may also be interested in a formal proof of the abstraction
theorem. It means that the deep embedding should be defined. As the
refinement is very close to \coq, this would thus require a large
effort.

\section{Related works and discussion}

Since the introduction of parametricity for system
F~\cite{DBLP:conf/ifip/Reynolds83,Wadler89}, it has been
extended to many logical systems based on Type Theory. Among others, we
can cite system $\Fw$ by Vytiniotis and
Weirich~\cite{VytiniotisW10} and a large subset of PTSs by
Bernardy \emph{et al.}~\cite{DBLP:conf/icfp/BernardyJP10,DBLP:conf/fossacs/BernardyL11}.
In all these presentations, no sort is impredicative, and parametricity
relations live either in a meta-logic or in a different sort than
propositions. To our knowledge, this is the first time parametricity
relations live in an impredicative sort representing propositions,
making them more usable in a system like \coq.

Bernardy \emph{et al.}~\cite{DBLP:conf/icfp/BernardyJP10} also explain two possible ways to
handle inductive definitions: one by translating induction principles,
and one by defining a new inductive data-type as the translation of the
initial data-type. Our approach is close to the second method proposed
by \cite{DBLP:conf/icfp/BernardyJP10}. We also show how to translate
fixpoint definitions, which are more common than inductive principles.

Parametricity and parts of the abstraction theorem have been formalized
for deep embeddings of logical systems in
\agda~\cite{DBLP:conf/icfp/BernardyJP10} and in
\coq~\cite{DBLP:conf/tlca/Atkey09,Atkey09b}. Our approach is different:
we do not want to have a formal proof of the abstraction theorem (in a
first step), but we want to have a practical tool that actually computes
results produced by the abstraction theorem. This does not compromise
soundness anyway, since the terms produced by this tool are type-checked
by \coq's kernel.

\section{Conclusion}

As we argue throughout the article, the system presented here
distinguishes clearly via typing which expressions will be computationally
meaningful after extraction. It allows us to define a notion of parametricity
for which relations lie in the sort of propositions. This opens up a new way
to define automatic tactics in interactive theorem provers based on Type
Theory.

Moreover it is known that parametricity and realizability seen as syntactic
constructions are closely related \cite{DBLP:conf/fossacs/BernardyL11}. That is why it seems
possible to build an internal realizability theory inside our framework. It
would permit to develop a similar tactic to prove automatically that program
extracted from any closed term will realize its own type. The user would then
be able to use this proof to show the correctness of his programs without
relying on the implementation of the extraction function.

Finally, it remains to understand why parametric relations do not fit in the
sort of proposition in presence of large elimination on non-small data types.
We conjecture that parametric relations for large inductive definitions are
not proof-irrelevant (in particular, they cannot be interpreted as
set-theoretical relations).

\subparagraph*{Acknowledgments}

The authors are particularly grateful to François Garillot and Georges
Gonthier who suggested the use of parametricity to obtain theorems from
free in the setting of algebra, and provided the stimulus for this work.
We also thank Assia Mahboubi for providing useful help about the spirit
of the Ssreflect library. We finally thank the anonymous reviewers for
their encouragements and constructive remarks.


\bibliography{biblio}

\end{document}